\documentclass[10pt,prd,longbibliography,superscriptaddress]{revtex4-2}
\usepackage{graphicx}
\usepackage{epstopdf,subfigure}
\usepackage{amsmath}

\usepackage{soul}
\usepackage{booktabs}
\usepackage{multirow}
\usepackage{microtype}
\usepackage{amsfonts}
\usepackage{float} 

\def\dis{\displaystyle}
\def\R{\mathbb R}
\newtheorem{Corollary}{Corollary}
\newtheorem{Theorem}{Theorem}
\newtheorem{proof}{Proof}
\begin{document}

\markboth{Agaoglou,Fe{\v c}kan, Posp\'i\v sil, Rothos \& Vakakis}{Periodically 
Forced Nonlinear Oscillatory Acoustic~Vacuum}

\title{Periodically 
Forced Nonlinear Oscillatory Acoustic~Vacuum}

\author{Makrina Agaoglou}
\address{Mathematical Institute of Slovak Academy of Sciences, \v{S}tef\'anikova 49, 814 73 Bratislava, Slovakia}
\address{Department of Mechanical Engineering, Faculty of Engineering, Aristotle University of Thessaloniki, 54124 Thessaloniki, Greece}

\author{Michal Fe\v{c}kan}
\address{Department of Mathematical Analysis and Numerical Mathematics, Comenius University in Bratislava,\\ Mlynsk\'a dolina, 842 48 Bratislava, Slovakia}
\address{Mathematical Institute of Slovak Academy of Sciences, \v{S}tef\'anikova 49, 814 73 Bratislava, Slovakia}

\author{Michal Posp\' i\v sil}
\address{Department of Mathematical Analysis and Numerical Mathematics, Comenius University in Bratislava,\\ Mlynsk\'a dolina, 842 48 Bratislava, Slovakia}
\address{Mathematical Institute of Slovak Academy of Sciences, \v{S}tef\'anikova 49, 814 73 Bratislava, Slovakia}

\author{Vassilis Rothos}
\address{Department of Mechanical Engineering, Faculty of Engineering, Aristotle University of Thessaloniki, 54124 Thessaloniki, Greece}

\author{Alexander F. Vakakis}
\address{Department of Mechanical Science and Engineering, University of Illinois at Urbana-Champaign, Urbana,IL 61801, USA}

\begin{abstract}
    In this work, we study the in-plane oscillations of a finite lattice of particles coupled by linear springs under distributed harmonic excitation. Melnikov-type analysis is applied for the persistence of periodic oscillations of a reduced system.
\end{abstract}

 
\maketitle
% Keywords
%

%\begin{document}
%%%%%%%%%%%%%%%%%%%%%%%%%%%%%%%%%%%%%%%%%%
%% Only for the journal Gels: Please place the Experimental Section after the Conclusions

%%%%%%%%%%%%%%%%%%%%%%%%%%%%%%%%%%%%%%%%%%

\section{Introduction}\label{sec1}

We analytically study the persistence of periodic oscillations for certain three-dimensional systems of ordinary differential equations (ODEs) with periodic perturbations and a slowly-varying variable. The considered ODEs are derived from a model of a finite lattice of particles coupled by linear springs under distributed harmonic excitation, which is described in detail in Section~\ref{sec2}. This model presents a low-energy nonlinear acoustic vacuum. We refer the reader for more motivations, further details and applications to~\cite{vakakis1,ZMSBV}. Following the computations of~\cite{vakakis1}, we consider just two modes in Section~\ref{sec3}, and we postpone higher modes investigation to our future paper, since another approach will be used. Melnikov analysis is demonstrated in Section~\ref{sec4} for finding conditions for the existence of periodic solutions for the perturbed ODEs corresponding to two modes. {More precisely, following~\cite{vakakis1}, we derive a three-dimensional periodically-perturbed system of ODEs with a slowly-varying variable. Then, we analyze an unperturbed autonomous system of ODEs to compute its family of periodic solutions by revising the results of~\cite{vakakis1} in more detail. Since we are interested in the persistence of periodic solutions for the perturbed ODEs, we compute the corresponding Melnikov functions by~\cite{WH}. Due to the difficulty of finding simple roots of these Melnikov functions explicitly, we outline an asymptotic approach for the location of some of them. Note that the simple roots of Melnikov functions predict the persistence and location of periodic solutions for perturbed ODEs. This is a novelty and a contribution of our paper.} Section~\ref{sec5} outlines possible future research along with summarizing our achievements in this paper.

\section{The Model}\label{sec2}

We consider a lattice consisting of $N$ identical particles coupled by identical linear springs (they~are un-stretched when the lattice is in the horizontal position) and executing in-plane oscillations (see~Figure~\ref{figure1}). Fixed boundary conditions and dissipative terms are imposed, and the transverse harmonic forces are also applied. The equations of motion can be expressed as follows,
\begin{equation}\label{eq.1}
\begin{array}{lll}
\dis m\frac{d^{2}u_{i}}{dt^{2}}+\left(T_{i} -\xi \frac{d\epsilon_i}{dt}\right)\cos
\phi_{i}-\left(T_{i+1}-\xi \frac{d\epsilon_{i+1}}{dt}\right)\cos \phi_{i+1}&=&0\\ [1.5ex]
\dis m\frac{d^{2}v_{i}}{dt^{2}}+\left(T_{i}-\xi \frac{d\epsilon_i}{dt}\right)\sin
\phi_{i}-\left(T_{i+1}-\xi \frac{d\epsilon_{i+1}}{dt}\right)\sin \phi_{i+1}-F_{i}
&=&0
\end{array}
\end{equation}
with $u_{i},v_{i}$ being the longitudinal and transversal displacements
of the $i$-th particle, respectively, $\phi_{i}$ the angle between the $i$-th spring and the horizontal direction, $\xi$ the damping coefficient,
$\epsilon _{i}=l'_{i}-l_{i}$ the deformation of the $i$-th spring,
$F_{i}$ the exciting transverse force and $m$ the mass of each particle
of the lattice. The tensile forces are proportional to the deformations of
the springs, and considering the geometry of the
deformed state of the lattice (see Figure~\ref{figure1}), one may write:
\begin{equation}
\begin{array}{lll}
\dis T_{i}&=&k(l'_{i}-l_{i}),\\
\dis \epsilon_{i}=l'_{i}-l_{i}&=&[(v_{i}-v_{i-1})^{2}+(l_{i}+u_{i}-u_{i-1})^{2}]^{1/2}-l_{i}
\end{array}
\end{equation}
with $l_{i}$ being the equilibrium length of the $i$-th spring (each spring has the same length) and $k$ the linear stiffness coefficient of each coupling spring. Introducing $\delta_{i}=\epsilon_{i}/l_{i}$, $s_{i}=u_{i}/l_{i}$, $w_{i}=v_{i}/l_{i}$, $c=\xi /(km)^{1/2}$, where $s_{i}$ and $w_{i}$ are the normalized axial and transverse displacements, and the ``slow'' time scale $\tau =(\frac{k}{m})^{1/2}t$, Equation (\ref{eq.1}) can be rewritten in normalized form:
\begin{equation}\label{eq.3}
\begin{array}{lll}
\dis \frac{d^{2}s_{i}}{d\tau^{2}}&=&\delta_{i}\cos \phi_{i}-\delta_{i-1}\cos \phi_{i-1}+c\delta'_{i}\cos \phi_{i}-c\delta'_{i-1}\cos \phi_{i-1}\\ [1.5ex]
\dis \frac{d^{2}w_{i}}{d\tau^{2}}&=&\delta_{i}\sin \phi_{i}-\delta_{i-1}\sin \phi_{i-1}+c\delta'_{i}\sin \phi_{i}-c\delta'_{i-1}\sin \phi_{i-1}+f_{i},
\end{array}
\end{equation}
where:
\begin{equation}
\begin{array}{rll}
\dis \delta_{i}&=&[(w_{i+1}-w_{i})^{2}+(1+s_{i+1}-s_{i})]^{1/2}-1,\\ [1.5ex]
\dis \delta'_{i}&=&\frac{(w_{i+1}-w_{i})(w_{i+1}'-w_{i}')+(1+s_{i+1}-s_{i})(s_{i+1}'-s_{i}')}{[(w_{i+1}-w_{i})^{2}+(1+s_{i+1}-s_{i})^{2}]^{1/2}}
\end{array}
\end{equation}
and:
\begin{equation}
\begin{array}{rrl}
\dis \cos \phi_{i}&=&\frac{1+s_{i+1}-s_{i}}{[(w_{i+1}-w_{i})^{2}+(1+s_{i+1}-s_{i})^{2}]^{1/2}},\\ [1.5ex]
\dis \sin \phi_{i}&=&\frac{w_{i+1}-w_{i}}{[(w_{i+1}-w_{i})^{2}+(1+s_{i+1}-s_{i})^{2}]^{1/2}},\\ [1.5ex]
\dis f_i&=&\frac{F_i}{k}.
\end{array}
\end{equation}

\begin{figure}[H]
\begin{center}
{\includegraphics[width=0.7\textwidth]{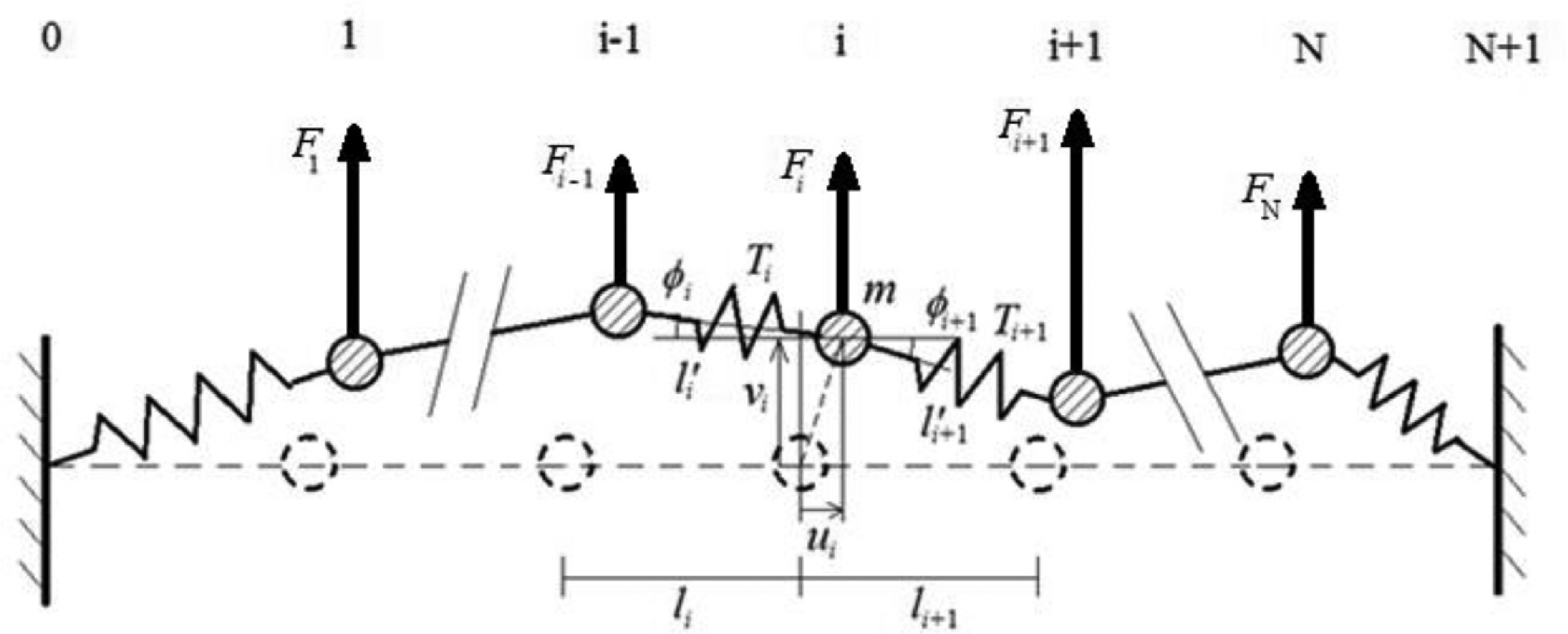}}
\end{center}
\caption{\normalsize Forced and damped lattice oscillating in the plane (see~\cite{ZMSBV})}.\label{figure1}
\end{figure} %%%4. Please replace with a sharper image.
%%%5. We notice the picture is cite from another paper (https://www.sciencedirect.com/science/article/pii/S0022509616308845). Due to Elsevier is non-open access journal, please provide the copyright permission of this paper to us soon.

The normalized system (\ref{eq.3}) is referred to as the ``exact lattice'' in the following sections.

According to the previous research~\cite{vakakis1}, when we introduce this system \eqref{eq.3} without extra transverse force and damping terms, an interesting feature is that in the low energy limit and under the assumption that the axial displacements are assumed to be an order of magnitude smaller compared to the transverse ones, it was shown that, correct to the leading order of approximation, the transverse oscillations decouple from the axial ones and are governed by the following reduced system of equations for predominantly transverse oscillations of the particles:
\begin{equation}\label{eq.6}
\begin{array}{ccc}
\dis \frac{d^{2}w_{i}}{d\tau^{2}}+2^{-1}(N+1)^{-1}\sum_{q=0}^{N}(w_{q+1}-w_{q})^{2}(2w_{i}-w_{i+1}-w_{i-1})=0,\\
\dis i=1,\cdots,N,\quad w_{0}(0)=w_{N+1}(0)=0.
\end{array}
\end{equation}

Then, the nearly linear axial oscillations are driven by the transverse responses (see~\cite{vakakis1,ZMSBV} for more details). Therefore, we focus our analysis like in~\cite{vakakis1} just on Equation~(\ref{eq.6}), which presents \hl{\em a low-energy nonlinear acoustic vacuum}, because in the absence of linear terms, it possesses zero speed of sound in the %%%6. is the italic necessary?
context of classical linear acoustics. What is more, it is notable that the existence of the strongly nonlocal multiplier $2^{-1}(N+1)^{-1}\sum_{q=0}^{N}(w_{q+1}-w_{q})^{2}$ indicates that the response of each particle is dependent (and hence, it is coupled) on the responses of all other particles. Equation~(\ref{eq.6}) admits \emph{N} exact nonlinear standing waves, or nonlinear normal modes (NNMs), in the form:
$$
w_{i}(\tau)=A_{p}(\tau) \sin \frac{\pi pi}{N+1},\quad i=1,\cdots,N
$$
for the $p$-th NNM, $1\leq p\leq N$, where $A_{p}(\tau)$ denotes the $p$-th modal amplitude. These, by construction, are mutually orthogonal, and there are no other NNMs in this system, nor any NNM bifurcations~\cite{vakakis1}. Substituting this NNM ansatz into Equation~(\ref{eq.6}) yields a set of $N$ uncoupled nonlinear equations governing the time-dependent amplitudes of the NNMs:
$$
A''_{p}(\tau)+(1/4)\omega_{p}^{4}A_{p}^{3}(\tau)=0
$$
with:
$$
\omega_{p}=2\sin\frac{\pi p}{2(N+1)},
$$
which is the $p$-th natural frequency of the corresponding linear system Equation~(\ref{eq.6}) and the prime denoting differentiation with respect to $\tau$. The exciting force, which is applied on each particle in the transverse direction, is expressed as:
$$
f_{i}=F_{p}\cos \omega_{p}\tau\sin \frac{pi\pi}{n+1}
$$
where $i=1,\cdots,N$, for the $p$-th NNM, $1\leq p\leq N$, and this exciting force includes NNMs in the form $\sin \frac{pi\pi}{n+1},\ i=1,\cdots,N,$ for the $p$-th NNM, $1\leq p\leq N$ and the $p$-th natural linear frequency $\omega_{p}$.

The frequency of the $p$-th NNM is tunable with the force and energy, and it also paves the way for nonlinear resonances between NNMs widely separated in the nonlinear spectrum, given that their energies tune their frequencies to satisfy certain rational relationships.

Summarizing, we can write \eqref{eq.6} as:
$$
\frac{d^{2}w}{d\tau^{2}}+2^{-1}(N+1)^{-1}\langle Mw,w\rangle Mw=0,
$$
where $w=[w_1,\cdots,w_N]\in\R^N$, $M$ is a symmetric matrix given by
$$
M=\begin{pmatrix} 2 & -1 & 0 & \cdots & 0\\
        -1 & 2 & -1& \cdots & 0\\
        \vdots & \vdots & \vdots & \vdots & \vdots \\
         0 & \cdots & -1 & 2 & -1\\
         0 &\cdots & 0 & -1 & 2
\end{pmatrix}
$$
and $\langle \cdot,\cdot\rangle$ is the standard scalar product on $\R^N$. The eigenvectors of $M$ are $\underline{\phi}_p=[\sin \frac{p\pi}{N+1},\cdots,\sin \frac{pN\pi}{N+1}]$ with the corresponding eigenvalues $\omega_p^2$, $1\le p\le N$. Moreover, it holds (see~\cite{GR}, p. 37):
$$
\begin{gathered}
\langle \underline{\phi}_p,\underline{\phi}_p\rangle=\sum_{i=1}^N\sin^2\frac{pi\pi}{N+1}=\frac{N+1}{2},\\
\langle \underline{\phi}_p,\underline{\phi}_k\rangle=\sum_{i=1}^N\sin\frac{pi\pi}{N+1}\sin\frac{ki\pi}{N+1}\\
=\frac{1}{2}\sum_{i=1}^N\left(\cos\frac{(p-k)i\pi}{N+1}-\cos\frac{(p+k)i\pi}{N+1}\right)=0,\quad p\ne k.
\end{gathered}
$$

The forced \eqref{eq.6} has the form:
\begin{equation}\label{eq.6b}
\frac{d^{2}w}{d\tau^{2}}+2^{-1}(N+1)^{-1}\langle Mw,w\rangle Mw=\sum_{p=1}^NF_p\cos \omega_{p}\tau\underline{\phi}_p.
\end{equation}

Therefore, considering the basis $\{\underline{\phi}_p\}_{p=1}^N$ of $\R^N$, we take $w(\tau)=\sum_{p=1}^NC_p(\tau)\underline{\phi}_p$ in \eqref{eq.6b} to get:
\begin{equation}\label{eq.6c}
C''_p(\tau)+\frac{1}{4}\left(\sum_{i=1}^NC_i^2(\tau)\omega_i^2\right)\omega_p^2C_p(\tau)=F_{p}\cos \omega_{p}\tau,\quad 1\le p\le N.
\end{equation}

Next, applying the coordinate transformation to \eqref{eq.6c}:
$$
A_{p}(\tau)=\frac{\omega_p}{2}C_{p}(\tau),\quad 1\le p\le N,
$$
we get:
$$
A''_p(\tau)+\left(\sum_{i=1}^NA_i^2(\tau)\right)\omega_p^2A_p(\tau)=\frac{F_{p}\omega_p}{2}\cos \omega_{p}\tau,\quad 1\le p\le N.
$$

\section{{Two}-Mode System}\label{sec3}
In this paper, we consider just two modes: $k$ and $p$, so we study the system:
\begin{equation}\label{eq.8}
\begin{array}{lll}
\dis A''_{k}(\tau)+[A^{2}_{k}(\tau)+A_{p}^{2}(\tau)]\omega_{k}^{2}A_{k}(\tau)+\epsilon\mu_1\cos (\omega_{k}\tau)&=&0\\ [1.5ex]
\dis A''_{p}(\tau)+[A^{2}_{k}(\tau)+A_{p}^{2}(\tau)]\omega_{p}^{2}A_{p}(\tau)+\epsilon\mu_2\cos (\omega_{p}\tau)&=&0,
\end{array}
\end{equation}
for $\epsilon\ne0$ small and parameters $\mu_1$, $\mu_2$. Using the transformation:
\begin{equation*}
\begin{array}{lll}
\dis \psi _{1}(\tau)&=&A_{k}'(\tau)+j\Omega A_{k}(\tau)\equiv \zeta_{1}(\tau )e ^{j\Omega \tau}\\ [1.5ex]
\dis \psi _{2}(\tau)&=&A_{p}'(\tau)+j\Omega A_{p}(\tau)\equiv \zeta_{2}(\tau )e ^{j\Omega \tau},
\end{array}
\end{equation*}
and performing an averaging approach like in~\cite{vakakis1}, Equation~(\ref{eq.8}) is modified to the form:
\begin{equation*}
\begin{array}{lll}
\dis \zeta_{1}'+\frac{j\Omega}{2}\zeta_{1}-\frac{j\omega_{k}^{2}}
{8\Omega ^{3}}(\zeta_{2}^{2}\zeta_{1}^{\ast}+2|\zeta_{2}|^{2}\zeta_{1}+3|
\zeta_{1}|^{2}\zeta_{1})+\epsilon\mu_1\cos (\omega _{k}\tau)e ^{-j\Omega \tau}
&=&0 \\ [1.5ex]
\dis \zeta_{2}'+\frac{j\Omega}{2}\zeta_{2}-\frac{j\omega_{p}^{2}}
{8\Omega ^{3}}(\zeta_{1}^{2}\zeta_{2}^{\ast}+2|\zeta_{1}|^{2}\zeta_{2}+3|
\zeta_{2}|^{2}\zeta_{2})+\epsilon\mu_2\cos (\omega _{p}\tau)e ^{-j\Omega \tau}
&=&0.
\end{array}
\end{equation*}

Introducing $\zeta_{i}=a_{i}e ^{j\beta _{i}}$ and $\Delta =\beta_{2}-\beta_{1}$, we get:
\begin{equation}\label{eq.28}
\begin{array}{rrl}
\dis a_{1}'+\frac{\omega_{k}^{2}}{8\Omega^{3}}a_{2}^{2}a_{1}\sin 2\Delta+\epsilon\mu_1\cos (\omega_{k}\tau)\cos(\Omega \tau +\beta _{1})&=& 0\\ [1.5ex]
\dis a_{2}'-\frac{\omega_{p}^{2}}{8\Omega^{3}}a_{1}^{2}a_{2}\sin 2\Delta+\epsilon\mu_2\cos (\omega_{p}\tau)\cos(\Omega \tau +\beta _{1}+\Delta)&=& 0\\ [1.5ex]
\dis \Delta '-\frac{\omega_{p}^{2}}{8\Omega^{3}}(3a_{2}^{2}+a_{1}^{2}\cos 2\Delta +2a_{1}^{2})+\frac{\omega_{k}^{2}}{8\Omega^{3}}(3a_{1}^{2}+a_{2}^{2}\cos 2\Delta +2a_{2}^{2})\\ [1.5ex]
\dis -\frac{\epsilon\mu_2}{a_{2}}\cos (\omega_{p}\tau)\sin (\Omega \tau +\beta _{1}+\Delta)
+\frac{\epsilon\mu_1}{a_{1}}\cos (\omega_{k}\tau)\sin (\Omega \tau +\beta _{1})&=&0
\end{array}
\end{equation}
where we consider $\beta_1$ as a constant parameter. Now, by introducing the coordinate transformations $a_{1}=(\frac{\rho}{\omega_{p}})\sin \theta$ and $a_{2}=(\frac{\rho}{\omega_{k}})\cos \theta$ into Equation~(\ref{eq.28}), we get:
$$
\begin{array}{rrl}
\dis \rho '=-\epsilon\mu_1\omega _{p}\sin \theta \cos(\omega_{k}\tau)\cos (\Omega \tau +\beta _{1})-\epsilon\mu_2\omega_{k}\cos\theta \cos(\omega_{p}\tau)\cos (\Omega \tau +\Delta +\beta _{1})\\ [1.5ex]
\dis \theta '+\frac{\rho ^{2}}{16\Omega ^{3}}\sin 2\theta \sin 2\Delta+ \epsilon\mu_1\omega_{p}\frac{\cos \theta}{\rho}\cos (\omega _{k}\tau)\cos (\Omega \tau +\beta _{1})\\ [1.5ex]
\dis -\epsilon\mu_2\omega_{k}\frac{\sin\theta}{\rho}\cos(\omega_{p}\tau)\cos (\Omega \tau +\Delta +\beta _{1})=0\\ [1.5ex]
\dis \Delta '-\frac{\rho ^{2}}{8\Omega^{3}}\left[\frac{3\omega_{p}^{2}}{\omega_{k}^{2}}\cos ^{2}\theta -\frac{3\omega_{k}^{2}}{\omega_{p}^{2}}\sin^{2}\theta-\cos 2\theta (2+\cos 2\Delta)\right]\\ [1.5ex]
\dis -\frac{\epsilon\mu_2\omega_{k}}{\rho \cos \theta}\cos (\omega _{p}\tau)\sin (\Omega \tau +\Delta +\beta _{1})+\frac{\epsilon\mu_1\omega_{p}}{\rho \sin \theta}\cos (\omega _{k}\tau)\sin (\Omega \tau +\beta _{1})=0.
\end{array}
$$

In the rest of the paper, we assume $\omega_p=\omega_k=P$ and $\Omega=kP$ for a natural number $k$, so we study the periodically-perturbed system:
\begin{equation}\label{eq.29}
\begin{array}{rrl}
\dis \rho '=-\epsilon\mu_1\sin \theta \cos(P\tau)\cos (kP\tau +\beta _{1})-\epsilon\mu_2\cos\theta \cos(P\tau)\cos (kP\tau +\Delta +\beta _{1})\\ [1.5ex]
\dis \theta '+\frac{\rho ^{2}}{16k^3P^{3}}\sin 2\theta \sin 2\Delta+ \epsilon\mu_1\frac{\cos \theta}{\rho}\cos (P\tau)\cos (kP\tau +\beta _{1})\\ [1.5ex]
\dis -\epsilon\mu_2\frac{\sin\theta}{\rho} \cos(P\tau)\cos (kP\tau +\Delta +\beta _{1})=0\\
\dis \Delta '-\frac{\rho ^{2}}{4k^3P^{3}}\cos 2\theta \sin ^{2}\Delta+\frac{\epsilon\mu_1}{\rho \sin \theta}\cos (P\tau)\sin (kP\tau +\beta _{1})\\ [1.5ex]
\dis -\frac{\epsilon\mu_2}{\rho \cos \theta}\cos (P\tau)\sin (kP\tau +\Delta +\beta _{1})=0
\end{array}
\end{equation}
where we scaled $P\mu_i\leftrightarrow \mu_i$, $i=1,2$. We may suppose:
$$
\mu_1^2+\mu_2^2=1.
$$

First, consider the unperturbed case where $\epsilon=0$, so the system:
\begin{equation}\label{eq.29a}
\begin{array}{rrl}
\dis \rho '=0\\
\dis \theta '+\frac{\rho ^{2}}{16k^3P^{3}}\sin 2\theta \sin 2\Delta=0\\ [1.5ex]
\dis \Delta '-\frac{\rho ^{2}}{4k^3P^{3}}\cos 2\theta \sin ^{2}\Delta=0.
\end{array}
\end{equation}

By introducing the temporal variable $\tau_{2}=\frac{\rho^{2}}{8k^3P^{3}}\tau$ in Equation~(\ref{eq.29a}), we get:
\begin{equation}\label{eq.32}
\begin{array}{rrr}
\dis \frac{d\theta}{d\tau_{2}}&=&-\frac{1}{2}\sin 2\theta \sin 2\Delta \\ [1.5ex]
\dis \frac{d\Delta}{d\tau_{2}}&=&2\cos 2\theta \sin ^{2}\Delta
\end{array}
\end{equation}
which is fully integrable and gives us the first integral $I=\sin 2\theta \sin \Delta =K=const.$ of the degenerate slow flow. If we consider the following initial conditions: $\theta (0)=\theta _{0}$, where \mbox{$0<\theta_{0}<\frac{\pi}{4}$} \mbox{and~$\Delta (0)=\pi /2$}, then we get $K=\sin 2\theta _{0}\in(0,1)$. To find exact solutions of \eqref{eq.32}, we first derive:
$$
\begin{gathered}
\frac{d\tau_{2}}{d\theta}=-\frac{2}{\sin 2\theta \sin 2\Delta}=-\frac{1}{K\cos\Delta}=\frac{1}{K\sqrt{1-\sin^2\Delta}}\\
=\frac{1}{K\sqrt{1-\frac{K^2}{\sin^22\theta}}}=\frac{1}{K}\frac{\sin2\theta}{\sqrt{\sin^22\theta-K^2}},
\end{gathered}
$$
for $\tau_2>0$ small, since \eqref{eq.32} gives $0<\theta<\frac{\pi}{2}$, $0<\Delta<\pi$ (see Figure~\ref{figure2}) and $\frac{d\Delta}{d\tau_{2}}(0)=2\cos 2\theta_0>0$, so~for $\tau_2>0$ small, we have $\Delta(\tau_2)>\frac{\pi}{2}$.

By using the formula~in \cite{GR} ((2.599.4) p. 205), we obtain:
$$
\tau_2=\int_{\theta_0}^\theta\frac{1}{K}\frac{\sin2\theta}{\sqrt{\sin^22\theta-K^2}}d\theta=\frac{1}{2K}\left(\frac{\pi}{2}-\sin ^{-1}\left(\frac{\cos 2\theta}{\sqrt{1-K^2}}\right)\right),
$$
which gives:
\begin{equation}\label{sol1}
\theta(\tau_2)=\frac{1}{2}\cos^{-1}\left(\sqrt{1-K^2}\cos(2K\tau_2)\right),
\end{equation}
recalling $0<\theta<\frac{\pi}{2}$. Of course, Formula \eqref{sol1} holds for all $\tau_2$, not just for small positive ones. Next, using \eqref{eq.32} and \eqref{sol1}, we have:
$$
\frac{d\Delta}{d\tau_{2}}=2\sqrt{1-K^2}\cos(2K\tau_2)\sin ^{2}\Delta,
$$
which can be easily solved to arrive at:
$$
\Delta(\tau_2)=\pi-\cot^{-1}\left(\frac{\sqrt{1-K^2}}{K}\sin(2K\tau_2)\right).
$$

\begin{figure}[H]
\centering
\includegraphics[width=10cm]{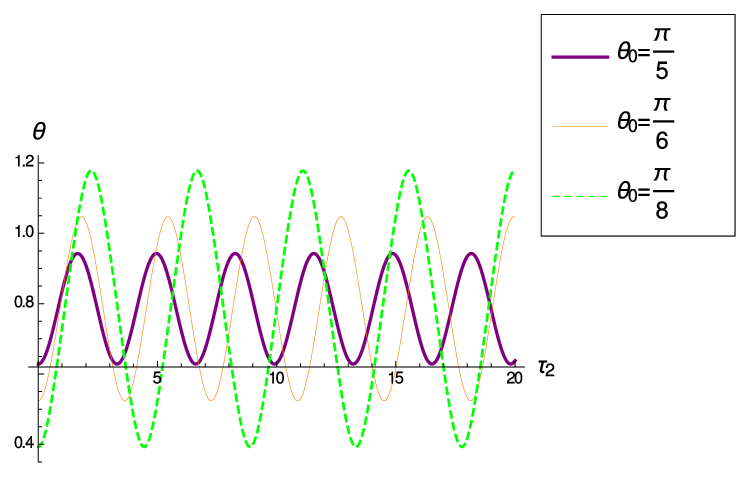} \\ \vspace{-6pt}
\includegraphics[width=10cm]{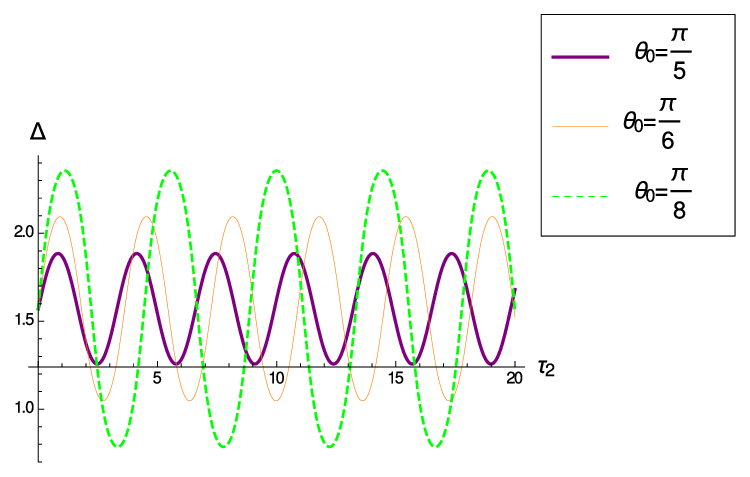}
\caption{{Top panel}: Graph for $\theta$ for different initial values of $\theta _{0}$.
{Bottom panel}: Graph for $\Delta$ for different initial values of $\theta _{0}$.}
\label{figure2}
\end{figure}%%%7. The black square has covered the line in the picture. Should be adjust?

Hence, the exact solution of the system (\ref{eq.32}) is given as:
$$
\begin{gathered}
\theta(\tau_2)=\frac{1}{2}\cos^{-1}\left(\sqrt{1-K^2}\cos(2K\tau_2)\right)\\
\Delta(\tau_2)=\pi-\cot^{-1}\left(\frac{\sqrt{1-K^2}}{K}\sin(2K\tau_2)\right)
\end{gathered}
$$
where the period is:
\begin{equation*}
T(\theta _{0})=\frac{2}{K}\int_{\theta _{0}}^{\frac{\pi}{2}-
\theta_{0}} \frac{d\theta}{\left(1-\frac{\sin ^{2}2\theta_{0}}{\sin ^{2}
2\theta}\right)^{1/2}}= \frac{\pi}{K}.
\end{equation*}

It is easy to verify the following symmetry property (see Figure~\ref{figure3}):
$$
\Delta\left(\tau_2+\frac{T}{2}\right)=\pi-\Delta(\tau_2),\quad \theta\left(\tau_2+\frac{T}{2}\right)=\frac{\pi}{2}-\theta(\tau_2).
$$

\begin{figure}[H]
\begin{center}
{\includegraphics[width=0.3\textwidth]{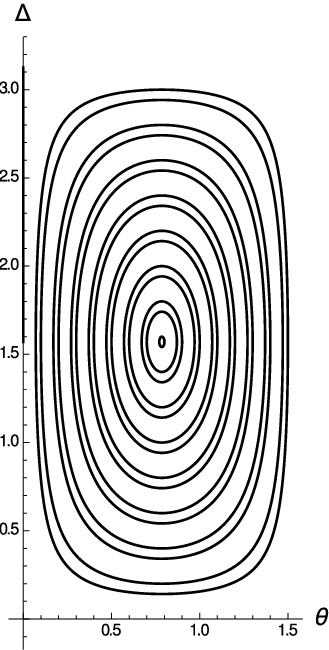}}
\end{center} \vspace{-6pt}
\caption{Orbits in the phase portrait of \eqref{eq.32}, where $(\theta,\Delta)\in [0,\frac{\pi}{2}]\times [0,\pi]$.}\label{figure3} %%%8. please provide the expalination for each color in the figure.
\end{figure}

Summarizing, the exact solution of the unperturbed \eqref{eq.29a} is the following:
\begin{equation*}
\begin{gathered}
\rho(\tau)=const.\\
\theta(\tau)=\frac{1}{2}\cos^{-1}\left(\sqrt{1-K^2}\cos\left(\frac{K\rho^{2}}{4k^3P^{3}}\tau\right)\right)\\
\Delta(\tau)=\pi-\cot^{-1}\left(\frac{\sqrt{1-K^2}}{K}\sin\left(\frac{K\rho^{2}}{4k^3P^{3}}\tau\right)\right)
\end{gathered}
\end{equation*}
with the period:
\begin{equation*}
T=\frac{8k^3P^{3}\pi}{K\rho^{2}}.
\end{equation*}

Consequently, for any:
\begin{equation*}
\rho>2kP^2\sqrt{k},
\end{equation*}
taking:
$$
K(\rho)=\frac{4k^3P^4}{\rho^2}\in(0,1),
$$
Equation \eqref{eq.29a} has the $T=2\pi/P$-periodic solution:

\begin{equation}\label{eq.34c}
\begin{gathered}
\theta(\rho,\tau)=\frac{1}{2}\cos^{-1}\left(\frac{\sqrt{\rho^4-16k^6P^8}}{\rho^2}\cos(P\tau)\right)\\
\Delta(\rho,\tau)=\pi-\cot^{-1}\left(\frac{\sqrt{\rho^4-16k^6P^8}}{4k^3P^4}\sin(P\tau)\right).
\end{gathered}
\end{equation}

\section{Melnikov Analysis for Periodic Oscillations}\label{sec4}
Writing \eqref{eq.29} as:
$$
\begin{array}{rrl}
\dis \theta '+\frac{\rho ^{2}}{16k^3P^{3}}\sin 2\theta \sin 2\Delta=\epsilon\rho^{-1}g_1(\rho,\theta,\Delta,\tau)\\ [1.5ex]
\dis \Delta '-\frac{\rho ^{2}}{4k^3P^{3}}\cos 2\theta \sin ^{2}\Delta=\epsilon\rho^{-1}g_2(\rho,\theta,\Delta,\tau)\\ [1.5ex]
\dis \rho '=\epsilon g_3(\rho,\theta,\Delta,\tau)
\end{array}
$$
and using~\cite{F}, (3.5.11), p. 111, with $\alpha=0$, and~\cite{H}, Lemma 2.5, p. 283, we compute the Melnikov~function:
$$
M(\beta_1,\rho)=\left(M_{1}(\beta_{1},\rho),M_{2}(\beta_{1},\rho)\right)
$$
as:
\begin{equation}\label{m1}
\begin{gathered}
M_{1}(\beta_{1},\rho)=\int_{0}^{T}\left(\frac{\partial I}{\partial \theta}\rho^{-1}g_{1}+\frac{\partial I}{\partial \Delta}\rho^{-1}g_{2}-\frac{\partial I}{\partial \theta}\frac{\partial \theta}{\partial \rho}g_{3}-\frac{\partial I}{\partial \Delta}\frac{\partial \Delta}{\partial \rho}g_{3}\right)d\tau\\
=\int_{0}^{T}\left(\frac{\partial I}{\partial \theta}\rho^{-1}g_{1}+\frac{\partial I}{\partial \Delta}\rho^{-1}g_{2}\right)d\tau -\frac{dK}{d\rho}\int_{0}^{T}g_{3}d\tau
\end{gathered}
\end{equation}
since differentiating by $\rho$ the identity:
$$
I(\theta(\rho,\tau),\Delta(\rho,\tau))=K(\rho),
$$
we get:
$$
\frac{\partial I}{\partial \theta}\frac{\partial \theta}{\partial \rho}+\frac{\partial I}{\partial \Delta}\frac{\partial \Delta}{\partial \rho}=\frac{dK}{d\rho},
$$
which is independent of $\tau$, and:
\begin{equation}\label{m2}
M_{2}(\beta_{1},\rho)=\int_{0}^{T}g_{3}d\tau.
\end{equation}

Formulas \eqref{m1} and \eqref{m2} are similar to~\cite{WH}, (2.7). We are looking for a simple zero of $M$, which is equivalent to considering:
\begin{equation*}
\begin{gathered}
\bar M(\beta_1,\rho)=\left(\bar M_{1}(\beta_{1},\rho),\bar M_{2}(\beta_{1},\rho)\right)\\
\bar M_{1}(\beta_{1},\rho)=\int_{0}^{T}\left(\frac{\partial I}{\partial \theta}g_{1}+\frac{\partial I}{\partial \Delta}g_{2}\right)d\tau,\quad \bar M_2(\beta_1,\rho)=\int_{0}^{T}g_{3}d\tau.
\end{gathered}
\end{equation*}

Since:
$$
g_i=\mu_1g_{i1}+\mu_2g_{i2},\quad i=1,2,3
$$
for:
$$
\begin{array}{l}
\dis g_{11}(\rho,\theta,\Delta,\tau)=-\cos \theta\cos (P\tau)\cos (kP\tau +\beta _{1})\\ [1.5ex]
\dis g_{12}(\rho,\theta,\Delta,\tau)=\sin\theta\cos(P\tau)\cos (kP\tau +\Delta +\beta _{1})\\ [1.5ex]
\dis g_{21}(\rho,\theta,\Delta,\tau)=-\frac{\cos (P\tau)}{\sin \theta}\sin (kP\tau +\beta _{1})\\ [1.5ex]
\dis g_{22}(\rho,\theta,\Delta,\tau)=\frac{\cos (P\tau)}{\cos \theta}\sin (kP\tau +\Delta +\beta _{1})\\ [1.5ex]
\dis g_{31}(\rho,\theta,\Delta,\tau)=-\sin \theta \cos(P\tau)\cos (kP\tau +\beta _{1})\\ [1.5ex]
\dis g_{32}(\rho,\theta,\Delta,\tau)=-\cos\theta \cos(P\tau)\cos (kP\tau +\Delta +\beta _{1})\\
\end{array}
$$
we get:
$$
\bar M_i(\beta_1,\rho)=\mu_1\bar M_{i1}(\beta_1,\rho)+\mu_2\bar M_{i2}(\beta_1,\rho),\quad i=1,2,
$$
for
$$
\bar M_{1j}(\beta_1,\rho)=\int_{0}^{T}\left(\frac{\partial I}{\partial \theta}g_{1j}+\frac{\partial I}{\partial \Delta}g_{2j}\right)d\tau,\quad \bar M_{2j}(\beta_1,\rho)=\int_{0}^{T}g_{3j}d\tau,\quad j=1,2.
$$

To solve:
$$
\mu_1\bar M_{i1}(\beta_1,\rho)+\mu_2\bar M_{i2}(\beta_1,\rho)=0,\quad i=1,2,
$$
we first solve the scalar equation:
\begin{equation}\label{m4}
\tilde M(\beta_1,\rho)=\bar M_{11}(\beta_1,\rho)\bar M_{22}(\beta_1,\rho)-\bar M_{12}(\beta_1,\rho)\bar M_{21}(\beta_1,\rho)=0
\end{equation}
to get its root $\beta_{1,0}$ and $\rho_{0}$. Then, we look for $\mu_{1,0}$ and $\mu_{2,0}$ with $\mu_{1,0}^2+\mu_{2,0}^2=1$ such that:
\begin{equation}\label{con1}
\begin{gathered}
\mu_{1,0}\bar M_{i1}(\beta_{1,0},\rho_0)+\mu_{2,0}\bar M_{i2}(\beta_{1,0},\rho_0)=0,\quad i=1,2\\
\det\begin{pmatrix}\mu_{1,0}\nabla\bar M_{11}(\beta_{1,0},\rho_0)^\top+\mu_{2,0}\nabla\bar M_{12}(\beta_{1,0},\rho_0)^\top\\
         \mu_{1,0}\nabla\bar M_{21}(\beta_{1,0},\rho_0)^\top+\mu_{2,0}\nabla\bar M_{22}(\beta_{1,0},\rho_0)^\top\end{pmatrix}\ne0.
\end{gathered}
\end{equation}

Summarizing, we have the following result.
\begin{Theorem}\label{th1}
If there are $\beta_{1,0}\in[0,2\pi)$, $\rho_{0}$ satisfying \eqref{m4}, $\mu_{1,0}$ and $\mu_{2,0}$ with $\mu_{1,0}^2+\mu_{2,0}^2=1$ solving \eqref{con1}, then for any $\mu_1$ near $\mu_{1,0}$ and $\mu_2$ near $\mu_{2,0}$ with $\mu_{1}^2+\mu_{2}^2=1$ and $\epsilon\ne0$ small, there are $\beta_1(\epsilon)$ near $\beta_{1,0}$ and $\rho(\epsilon)$ near $\rho_{0}$ such that \eqref{eq.29} with $\beta_1=\beta_1(\epsilon)$ and $\rho=\rho(\epsilon)$ has a $T=2\pi/P$-periodic solution near \eqref{eq.34c} with $\rho=\rho(\epsilon)$.
\end{Theorem}
Note:
$$
\begin{gathered}
\bar M_{11}(\beta_1,\rho)=-\int_0^T\Big(2\cos(2\theta(\rho,\tau))\sin\Delta(\rho,\tau)\cos \theta(\rho,\tau)\cos (P\tau)\cos (kP\tau +\beta _{1})\\
+2\cos\theta(\rho,\tau)\cos\Delta(\rho,\tau)\cos (P\tau)\sin (kP\tau +\beta _{1}))\Big)d\tau\\
\bar M_{12}(\beta_1,\rho)=\int_0^T\Big(2\cos(2\theta(\rho,\tau))\sin\Delta(\rho,\tau)\sin\theta(\rho,\tau)\cos(P\tau)\cos (kP\tau +\Delta(\rho,\tau) +\beta _{1})\\
+2\sin\theta(\rho,\tau)\cos\Delta(\rho,\tau)\cos (P\tau)\sin (kP\tau +\Delta(\rho,\tau) +\beta _{1})\Big)d\tau\\
\bar M_{21}(\beta_1,\rho)=-\int_0^T\sin \theta(\rho,\tau) \cos(P\tau)\cos (kP\tau +\beta _{1})d\tau\\
\bar M_{22}(\beta_1,\rho)=-\int_0^T\cos\theta(\rho,\tau) \cos(P\tau)\cos (kP\tau +\Delta(\rho,\tau) +\beta _{1})d\tau.
\end{gathered}
$$

Next, taking $\rho\to\infty$ in \eqref{eq.34c}, we obtain:
\begin{equation}\label{eq.34d}
\begin{gathered}
\theta(\infty,\tau)=\begin{cases}\dis\frac{P\tau}{2}\quad \tau\in\left[0,\frac{\pi}{P}\right],\\ \\[-2ex]
               \dis\pi-\frac{P\tau}{2}\quad \tau\in\left[\frac{\pi}{P},\frac{2\pi}{P}\right],\end{cases}\\
\Delta(\infty,\tau)=\begin{cases}\dis\pi\quad \tau\in\left(0,\frac{\pi}{P}\right),\\ \\[-2ex]
               \dis0\quad \tau\in\left(\frac{\pi}{P},\frac{2\pi}{P}\right),\\ \\[-2ex]
               \dis\frac{\pi}{2}\quad \tau\in\left\{0,\frac{\pi}{P},\frac{2\pi}{P}\right\}, \end{cases}\\
\frac{\partial\theta}{\partial\rho}(\infty,\tau)=0,\quad \frac{\partial\Delta}{\partial\rho}(\infty,\tau)=0.
\end{gathered}
\end{equation}

Hence:
$$
\begin{gathered}
\bar M_{11}(\beta_1,\infty)=\int_0^{\frac{2\pi}{P}}2\cos\frac{P\tau}{2}\cos (P\tau)\sin(kP\tau +\beta _{1})d\tau\\
=\frac{16k(4k^2-5)\cos\beta_1}{(16 k^4-40 k^2+9)P}\\
\bar M_{12}(\beta_1,\infty)=\int_0^{\frac{2\pi}{P}}2\sin\frac{P\tau}{2}\cos (P\tau)\sin (kP\tau+\beta _{1})\Big)d\tau\\
=-\frac{8(4 k^2+3)\sin\beta_1}{(16 k^4-40 k^2+9)P}\\
\bar M_{21}(\beta_1,\infty)=-\int_0^{\frac{2\pi}{P}}\sin\frac{P\tau}{2}\cos(P\tau)\cos (kP\tau +\beta _{1})d\tau\\
=\frac{4(4 k^2+3)\cos\beta_1}{(16 k^4-40 k^2+9)P}\\
\bar M_{22}(\beta_1,\infty)=\int_0^{\frac{2\pi}{P}}\cos\frac{P\tau}{2}\cos(P\tau)\cos (kP\tau+\beta _{1})d\tau\\
=-\frac{8k(4 k^2-5)\sin\beta_1}{(16k^4-40 k^2+9)P}.
\end{gathered}
$$

Then, \eqref{m4} gives as $\rho\to\infty$,
$$
\tilde M(\beta_1,\infty)=-\frac{16\sin2\beta_1}{(4k^2-9)P^2}
$$
with asymptotic solutions $\beta_{1,0}^\infty$ satisfying either $\sin\beta_{1,0}^\infty=0$ or $\cos\beta_{1,0}^\infty=0$. The asymptotic equation of~\eqref{con1} is as follows:
$$
\begin{gathered}
0=\mu_{1,0}\bar M_{11}(\beta_{1,0}^\infty,\infty)+\mu_{2,0}\bar M_{12}(\beta_{1,0}^\infty,\infty)\\ \\[-2ex]
=\mu_{1,0}\frac{16k(4k^2-5)\cos\beta_{1,0}^\infty}{(16 k^4-40 k^2+9)P}-\mu_{2,0}\frac{8(4 k^2+3)\sin\beta_{1,0}^\infty}{(16 k^4-40 k^2+9)P},
\end{gathered}
$$
which has solutions: $\mu_{1,0}^\infty=0$ and $\mu_{2,0}^\infty=\pm1$ when $\sin\beta_{1,0}^\infty=0$ or $\mu_{1,0}^\infty=\pm1$ and $\mu_{2,0}^\infty=0$ when $\cos\beta_{1,0}^\infty=0$. However, $\frac{\partial \bar M_{ij}}{\partial\rho}(\beta_1,\infty)=0$ for $i,j\in\{1,2\}$, so Theorem~\ref{th1} cannot be applied directly. Note we consider just the first asymptotic equation of \eqref{con1}, since the second one is a scalar multiple of the first~one.

On the other hand, following the method of~\cite{F}, p. 111, we get the following result.

\begin{Corollary}\label{co1}
For any $\rho>0$ sufficiently large and $\epsilon\ne0$ sufficiently small, there is $\mu_1(\rho,\epsilon)$, $\mu_2(\rho,\epsilon)$ with $\mu_1^2(\rho,\epsilon)+\mu_2^2(\rho,\epsilon)=1$ and $\beta_1(\rho,\epsilon)$ such that \eqref{eq.29} has a $T$-periodic solution near \eqref{eq.34d} for $\epsilon\ne0$ small and either $\mu_1(\infty,0)=0$, $\mu_2(\infty,0)=\pm1$, $\sin\beta_1(\infty,0)=0$ or $\mu_1(\infty,0)=\pm1$, $\mu_2(\infty,0)=0$, $\cos\beta_1(\infty,0)=0$.
\end{Corollary}
\begin{proof}
The bifurcation equation has the form (see~\cite{F}, p. 111):
$$
\mu_1\bar M_{i1}(\beta_1,\rho)+\mu_2\bar M_{i2}(\beta_1,\rho)=O(\epsilon),\quad i=1,2,\quad \mu_1^2+\mu_2^2=1
$$
i.e.,
\begin{equation}\label{bif1}
\sin\Gamma\bar M_{i1}(\beta_1,\rho)+\cos\Gamma\bar M_{i2}(\beta_1,\rho)=O(\epsilon),\quad i=1,2
\end{equation}
for $\mu_1=\sin\Gamma$ and $\mu_2=\cos\Gamma$. For $\rho=\infty$ and $\epsilon=0$, \eqref{bif1} takes the form:
\begin{equation}\label{bif2}
\begin{gathered}
\frac{16k(4k^2-5)}{(16 k^4-40 k^2+9)P}\sin\Gamma\cos\beta_1-\frac{8(4 k^2+3)}{(16 k^4-40 k^2+9)P}\cos\Gamma\sin\beta_1=0\\
\frac{4(4 k^2+3)\cos\beta_1}{(16 k^4-40 k^2+9)P}\sin\Gamma\cos\beta_1-\frac{8k(4 k^2-5)}{(16k^4-40 k^2+9)P}\cos\Gamma\sin\beta_1=0.
\end{gathered}
\end{equation}

The determinant of \eqref{bif2} is $-\frac{32}{(4k^2-9)P^2}\ne0$, so \eqref{bif2} has the only solutions:
$$
\sin\Gamma_\infty=\sin\beta_{1,0}^\infty=0,\quad \cos\Gamma_\infty=\cos\beta_{1,0}^\infty=0.
$$

The determinants of Jacobians of \eqref{bif2} at these zeros are as follows:
$$
\begin{gathered}
\det\begin{pmatrix} \frac{16k(4k^2-5)}{(16 k^4-40 k^2+9)P}\cos\Gamma_\infty\cos\beta_{1,0}^\infty&-\frac{8(4 k^2+3)}{(16 k^4-40 k^2+9)P}\cos\Gamma_\infty\cos\beta_{1,0}^\infty\\
\frac{4(4 k^2+3)\cos\beta_{1,0}^\infty}{(16 k^4-40 k^2+9)P}\cos\Gamma_\infty\cos\beta_{1,0}^\infty &-\frac{8k(4 k^2-5)}{(16k^4-40 k^2+9)P}\cos\Gamma_\infty\cos\beta_{1,0}^\infty
\end{pmatrix}\\
=-\frac{32\cos^2\Gamma_\infty\cos^2\beta_{1,0}^\infty}{(4k^2-9)P^2}\ne0
\end{gathered}
$$
and:
$$
\begin{gathered}
\det\begin{pmatrix} \frac{8(4 k^2+3)}{(16 k^4-40 k^2+9)P}\sin\Gamma_\infty\sin\beta_{1,0}^\infty&-\frac{16k(4k^2-5)}{(16 k^4-40 k^2+9)P}\sin\Gamma_\infty\sin\beta_{1,0}^\infty&\\
\frac{8k(4 k^2-5)}{(16k^4-40 k^2+9)P}\sin\Gamma_\infty\sin\beta_{1,0}^\infty&-\frac{4(4 k^2+3)\cos\beta_{1,0}^\infty}{(16 k^4-40 k^2+9)P}\sin\Gamma_\infty\sin\beta_{1,0}^\infty
\end{pmatrix}\\
=\frac{32\sin^2\Gamma_\infty\sin^2\beta_{1,0}^\infty}{(4k^2-9)P^2}\ne0,
\end{gathered}
$$
respectively. Hence, the zeroes $\Gamma_\infty$ and $\beta_{1,0}^\infty$ are simple, so we can apply the implicit function theorem to get the result. The proof is complete.
\end{proof}

\section{Discussion}\label{sec5}

Melnikov analysis is applied for the persistence of periodic oscillations for periodically-perturbed systems of ODEs with a slowly-varying variable. The ODEs are obtained from a model of a finite lattice of particles coupled by linear springs under distributed harmonic excitation, which presents a low-energy nonlinear acoustic vacuum (see~\cite{vakakis1}), but we consider just two modes. We extend the study of~\cite{vakakis1} to a problem with small exciting harmonic forces. We apply an analytical method based on derivation of Melnikov functions and then on the location of their simple roots. Melnikov functions are derived by using the approach of~\cite{F,WH} developed for slowly-varying ODEs. Since the Melnikov functions are rather complicated, we follow an asymptotic way for solving the corresponding Melnikov equations. It~would be nice to solve these Melnikov equations numerically for finding other simple roots, which is postponed to our next research. These roots determine and locate periodic solutions of the periodically-perturbed systems of ODEs derived from the two-mode low-energy nonlinear acoustic vacuum system. Moreover, our next investigation will be also to consider higher numbers of modes represented by a system of ODEs in \eqref{eq.6c}. The method used will be different from that in this paper, since it will be based on the results of Section 3.3 of~\cite{F}. Note that higher modes of \eqref{eq.3} were numerically studied in~\cite{ZMSBV}, which is another challenge for our study.\\

Authors Contributions: The authors contributed equally to this work. %%%9. Please provide a detailed description of each author's contribution: For research articles with several authors, a short paragraph specifying their individual contributions must be provided. The following statements should be used “Conceptualization, X.X. and Y.Y.; Methodology, X.X.; Software, X.X.; Validation, X.X., Y.Y. and Z.Z.; Formal Analysis, X.X.; Investigation, X.X.; Resources, X.X.; Data Curation, X.X.; Writing-Original Draft Preparation, X.X.; Writing-Review & Editing, X.X.; Visualization, X.X.; Supervision, X.X.; Project Administration, X.X.; Funding Acquisition, Y.Y.”, please turn to the CRediT taxonomy for the term explanation. Authorship must be limited to those who have contributed substantially to the work reported.

%%%%%%%%%%%%%%%%%%%%%%%%%%%%%%%%%%%%%%%%%%
Funding:M.A. is supported by the National Scholarship Programme of the Slovak Republic for the Support of Mobility of Students, Ph.D. Students, University Teachers, Researchers and Artists. M.F. and M.P. are supported by the Slovak Research and Development Agency (Grant Number APVV-14-0378) and the Slovak Grant Agency VEGA%%%10.define if appropriate
 (Grant Numbers
2/0153/16 and 1/0078/17).

%%%%%%%%%%%%%%%%%%%%%%%%%%%%%%%%%%%%%%%%%%
%\acknowledgments{In this section you can acknowledge any support given which is not covered by the author contribution or funding sections. This may include administrative and technical support, or donations in kind (e.g. materials used for experiments).}

%%%%%%%%%%%%%%%%%%%%%%%%%%%%%%%%%%%%%%%%%%

%%%%%%%%%%%%%%%%%%%%%%%%%%%%%%%%%%%%%%%%%%

\begin{thebibliography}{999}
%4
\bibitem{vakakis1} Manevich, L.I.;  Vakakis, A.F. Nonlinear oscillatory acoustic vacuum. \textit{SIAM J. Appl. Math.} {\bf 2014}, \emph{74}, 1742--1762.
%6
\bibitem{ZMSBV} Zhang, Z.; Manevitch, L.I.; Smirnov, V.; Bergman, L.; Vakakis, A.F. Extreme nonlinear energy exchanges in a geometrically nonlinear lattice oscillating
in the plane. \textit{J. Mech. Phys. Solids} {\bf 2018}, \emph{110}, 1--20.
%5
\bibitem{WH} Wiggins, S.;  Holmes, P. Periodic orbits in slowly varying oscillators. \textit{SIAM J. Math. Anal.} {\bf 1987}, \emph{18}, 592--611.

%2
\bibitem{GR} Gradshteyn, I.S.;  Ryzhik, I.M. \textit{Table of Integrals, Series, and Products}, 7th ed.; Academic Press: \mbox{Cambridge, MA, USA}, 2007.
%1
\bibitem{F} Fe\v ckan, M. \textit{Topological Degree Approach to Bifurcation Problems}; Springer: Berlin, Germany, 2008.

%3
\bibitem{H} Hale, J.K. \textit{Ordinary Differential Equations}, 2nd ed.; Robert E. Krieger Publishing Company, Inc.: {Malabar, FL, USA}, 1980.



\end{thebibliography}
\end{document}